\documentclass[sigconf]{acmart}
\AtBeginDocument{%
  }

\setcopyright{acmlicensed}
\copyrightyear{2025}
\acmYear{2025}
\acmDOI{XXXXXXX.XXXXXXX}

\definecolor{MyGreen}{rgb}{0.200,0.500,0.200}
\renewcommand{\emph}[1]{{\color{MyGreen}{\em #1}}}

\acmConference[Conference acronym 'XX]{Make sure to enter the correct
  conference title from your rights confirmation emai}{June 03--05,
  2018}{Woodstock, NY}
\acmISBN{978-1-4503-XXXX-X/18/06}

\usepackage{amsmath}
\usepackage{enumerate,enumitem}

\usepackage{pgf,tikz,pgfplots}
\usetikzlibrary{positioning}
\usepackage{multirow}

\usepackage{algorithm}
\usepackage[noend]{algpseudocode}

\newcommand{\OurProb}{\texttt{LDS}}

\RequirePackage{hyperref}
\hypersetup{
  colorlinks=true,
  linkcolor=green!30!black,
  linktoc=all,
  citecolor=violet,
  bookmarksnumbered=true,
  pdfproducer={},
  pdfinfo={CreationDate={2024},ModDate={2024}},
  pdfauthor={},
  pdfkeywords={},
  pdfsubject={},
  pdfcreator={},
  pdflang={}
}

\begin{document}
%\hideACM

\title{Limiting Disease Spreading in Human Networks}

\author{Gargi Bakshi}
 \affiliation{%
   \institution{Indian Institute of Technology Bombay}
   \country{}}
 \email{gargibakshi@cse.iitb.ac.in}

\author{Sujoy Bhore}
 \affiliation{%
   \institution{Indian Institute of Technology Bombay}
   \country{}}
 \email{sujoy@cse.iitb.ac.in}

\author{Suraj Shetiya}
 \affiliation{%
   \institution{Indian Institute of Technology Bombay}
   \country{}}
 \email{surajs@cse.iitb.ac.in}

\begin{abstract}
The outbreak of a pandemic, such as COVID-19, causes major health crises worldwide. Typical measures to contain the rapid spread usually include effective vaccination and strict interventions (Nature Human Behaviour, 2021). Motivated by such circumstances, we study the problem of limiting the spread of a disease over a social network system.

In their seminal work (KDD 2003), Kempe, Kleinberg, and Tardos introduced two fundamental diffusion models, the \emph{linear threshold} and \emph{independent cascade}, for the influence maximization problem. In this work, we adopt these models in the context of disease spreading and study effective vaccination mechanisms. Our broad goal is to limit the spread of a disease in human networks using only a limited number of vaccines. However, unlike the influence maximization problem, which typically does not require spatial awareness, disease spreading occurs in spatially structured population networks. Thus, standard \emph{Erd\H{o}s-R\'{e}nyi graphs} do not adequately capture such networks. To address this, we study networks modeled as \emph{generalized random geometric graphs}, introduced in the seminal work of Waxman (IEEE J. Sel. Areas Commun. 1988).

We show that for disease spreading, the optimization function is neither \emph{submodular} nor \emph{supermodular}, in contrast to influence maximization, where the function is \emph{submodular}. Despite this intractability, we develop novel algorithms leveraging local search and greedy techniques, which perform exceptionally well in practice. We compare them against an exact ILP-based approach to further demonstrate their robustness. Moreover, we introduce an iterative rounding mechanism for the relaxed LP formulation. Overall, our methods establish tight trade-offs between efficiency and approximation loss.

\end{abstract}
\setcopyright{none}
\settopmatter{printacmref=false}

\maketitle

\section{Introduction}\label{sec:introduction}

The rapid spread of infectious diseases across populations is a global challenge with profound public health~\cite{gubler2002epidemic,zimmet2001global}, economic~\cite{hays2005epidemics,madhav2017pandemics,meltzer1999economic}, and social implications~\cite{yip2010impact,strong1990epidemic}. For instance, the recent outbreaks of epidemics, such as COVID-19, Ebola, and Zika have underscored the importance of understanding how diseases propagate through interconnected human networks~\cite{sampathkumar2016zika, weaver2016zika}. These pandemics have demonstrated that modern societies are increasingly vulnerable to the fast-paced transmission of pathogens, exacerbated by the dense and complex networks of social interactions. Containing such outbreaks requires more than traditional public health measures. It requires an understanding of the underlying processes that govern disease spread in interconnected populations. As such, there is an urgent need for computational tools that can simulate and predict disease dynamics over complex network structures.

Traditional epidemiological models, such as the \emph{Susceptible Infected Recovered} (SIR) and \emph{Susceptible Exposed Infected Recovered} (SEIR) frameworks (\cite{kermack1927contribution}), were initially developed to describe disease spread in homogeneous populations. However, these models are limited in their capacity to capture the complexities of real-world transmission, where individuals interact in diverse and often heterogeneous ways, typically within networks of variable connectivity. Recent research has highlighted that the topology of these networks—whether social, geographic, or biological—plays a critical role in shaping epidemic dynamics. For instance, highly connected individuals, or ``hubs'', can act as \emph{super-spreaders}, significantly accelerating the transmission process, while network structures such as community clustering can either impede or facilitate disease spread, depending on the nature of the disease and its mode of transmission (\cite{barabasi1999emergence, bedson2021review, newman2003structure}).

Globally, there has been a significant increase in urban populations, particularly in developing and developed nations. This demographic shift exacerbates the spread of diseases within these regions. Traditional network (graph) generation models, such as \emph{Erd\H{o}s-R\'{e}nyi (ER)}~\cite{erdos1960evolution} and \emph{Barab\'{a}si-Albert (BA)}~\cite{albert2002statistical}, have been used to create graph topologies representing real-world scenarios. ER graphs have been employed to model social networks~\cite{seshadhri2012community}, power grids~\cite{shahraeini2023modified}, and biological networks~\cite{przulj2004graph}, while BA graphs produce random scale-free networks through preferential attachment. BA graphs are often used to simulate structures like the World Wide Web, citation networks, and social networks. These models, however, do not consider the \emph{spatial locations} of nodes and their \emph{proximity} to each other, which is a crucial aspect of disease propagation.
Therefore, in this work, we focus on networks based on \emph{generalized random geometric graphs}, introduced in the seminal work by Waxman~\cite{waxman1988routing}, which addresses this limitation.

Furthermore, the influence propagation models introduced by Kempe, Kleinberg, and Tardos in their seminal work~\cite{kempe2003maximizing} provide a robust theoretical foundation for understanding diffusion processes, including any form of spreading mechanisms over networks. 
Given the structural parallels between influence propagation and disease transmission in human networks, it is both intuitive and compelling to apply these models to the context of epidemic spread. Specifically, their influence maximization framework, which captures the diffusion of information or behaviors through a network, highlights the role of key nodes (\emph{super-spreaders}) in amplifying contagion. This closely mirrors disease dynamics in complex networks, where transmission is shaped by node connectivity and local interactions (see~\cite{pastor2001epidemic, fortunato2010community, pastor2015epidemic}).
Hence, applying influence maximization techniques to disease spread offers the potential to identify critical nodes whose intervention can significantly impact the course of an epidemic (see~\cite{chen2009efficient, lu2016vital}). 
Moreover, one potential hope is to leverage optimization methods from the field of network science, such as those for submodular maximization, it might be possible to develop efficient, data-driven strategies for epidemic control and resource allocation (see~\cite{filmus2014monotone}).

\subsection{Our Contributions.}

Surprisingly, we find that in the context of disease spreading, the optimization function is neither \emph{submodular} nor \emph{supermodular} (Section~\ref{sec:hardness-results}), unlike in influence maximization, where the function is \emph{submodular}. Hence, we can not expect to provide greedy strategies with provable guarantees like Kempe et al.~\cite{kempe2003maximizing}. 
Furthermore, the computational intractability of the disease spreading problem can be easily established by a reduction from the well-known \texttt{Minimum Set Cover}\footnote{Given a universe $U$ and a collection of subsets $S$ of $U$, and the goal is to find the smallest number of subsets from $S$ that
cover all the elements in $U$.} problem. Specifically, we construct a bipartite graph where the objective is to select a minimum number of nodes from one partition such that the union of their neighborhoods covers\footnote{In this context, ``cover'' means that the union of the selected nodes' neighborhoods encompasses all nodes in the other partition.} all nodes in the opposite partition. Clearly, if it is possible to find such a cover, then we can find a solution for our problem and vice-versa. 
For conciseness, we omit the full reduction, as it is relatively straightforward. Hence, we primarily focus on developing efficient and robust methods. To that end, we devise the following algorithmic approaches (Section~\ref{sec:algorithms}). 

\begin{itemize}[leftmargin=*]
    \item \emph{Sampling based strategy}: We propose a general sampling based strategy (Section~\ref{subsec:sampling-approach}) to avoid enumerating all $2^{|E|}$ topologies. Sampling topologies provide us trade-off between runtime and quality of results.
    \item \emph{Linear Program with Binary variables}: We define a linear programming based approach (Section~\ref{subsec:ilpr-relaxation}) that models the vertices to be vaccinated as binary variables and the disease spreading as constraints. For a given set of topologies, the solution provided by this approach is optimal, but with a high running time.
    \item \emph{Relaxed Linear Program with Relaxations}: To overcome the high runtime, we relax the Binary Linear Program and define two rounding strategies - \texttt{TKR} and \texttt{IRP} (Section~\ref{subsec:ilpr-relaxation}). While the Top-$k$ Rounding procedure (\texttt{TKR}) considers the top-$k$ candidates for the rounding process, \texttt{IRP} Iterative rounding procedure uses the most probable vertex in each iteration to be added to the set of vaccinated vertices. 
    \item \emph{Greedy}: We also propose a Greedy approach (Section~\ref{subsec:greedy}), which initially considers an empty set of vertices and in each iteration adds one item to the set to be vaccinated. 
    \item \emph{Local Search}: Local Search has been used effectively to solve numerous problems in graph theory. In our problem, we define the notion of the neighborhood of a set of vertices and propose a local search based heuristic (Section~\ref{subsec:local-search}) to improve the results obtained from the greedy approach.
\end{itemize}

We substantiate these algorithms with rigorous runtime analysis and extensive empirical evaluations.

\section{Related Work}\label{sec:related_work}

The study of disease spread in human networks has been extensively explored through both epidemiological and network-based models. Classical models such as SIR and SEIR~\cite{kermack1927contribution} provide a foundational framework but assume homogeneous mixing, limiting their applicability to real-world heterogeneous networks. Subsequent research has emphasized the role of network topology in shaping disease transmission, demonstrating that highly connected nodes (\emph{super-spreaders}) can accelerate outbreaks~\cite{barabasi1999emergence, bedson2021review, newman2003structure}. Recent epidemics such as COVID-19, Ebola, and Zika have further highlighted the role of network-driven transmission and its broad public health, economic, and social consequences~\cite{gubler2002epidemic, zimmet2001global, hays2005epidemics, madhav2017pandemics, meltzer1999economic, yip2010impact, strong1990epidemic}. Traditional graph-based models, including Erd\H{o}s–R\'{e}nyi (ER) and Barabási–Albert (BA) networks~\cite{erdos1960evolution, albert2002statistical}, have been widely used to simulate transmission patterns in domains such as social networks~\cite{seshadhri2012community}, power grids~\cite{shahraeini2023modified}, and biological systems~\cite{przulj2004graph}. However, these models often overlook spatial proximity, a crucial factor in disease spread. Generalized random geometric graphs (RGGs) address this limitation by incorporating spatial constraints~\cite{waxman1988routing}.

In parallel, influence propagation models~\cite{kempe2003maximizing} have provided theoretical insights into diffusion processes. Given their structural parallels with epidemic spread, these models have been applied to identify key nodes for disease containment~\cite{pastor2001epidemic, fortunato2010community, pastor2015epidemic}. Recent work has leveraged submodular optimization techniques~\cite{chen2009efficient, lu2016vital, filmus2014monotone} to develop efficient, targeted intervention strategies, highlighting the potential for computational approaches in epidemic control.
\section{Preliminaries}\label{sec:preliminaries}

\subsection{Notations}\label{subsec:labels}

We briefly summarise the intended behaviour of the model and thus motivate our choice of the model.
Our choice of the model hinges on the process of disease spreading.
A set of individuals who may be part of a community or vicinity are chosen to be part of the network.
As these individuals interact in the physical world, these interactions need to be considered in the dynamics of the disease-spreading model.
Moreover, some interactions might last for a longer time than others and hence, these individuals are more susceptible to infection.
If one person from a community is infected, then the infection could spread due to interactions. \emph{Stronger the frequency or duration of interaction, the greater the chance of contracting a disease for the people interacting}.

Based on these constraints, we model the spread of disease similar to that of influence spreading networks (independent cascade and linear threshold models) introduced by Kempe et al.~\cite{kempe2003maximizing}.

\smallskip
\noindent\textbf{Network model}: Consider a graph $G(V, E)$, where the nodes $V$ represent a set of individuals and the edges $E$ represent the connections between people. 
We model the extent of an individual $v_j$ contracting the disease via a neighbor $v_i$ based on the strength of the interactions between them. Moreover, individuals with less optimal hygiene practices may be more susceptible to contracting a disease, which makes the interaction ``\emph{non-uniform}''. The number of nodes $|V|$ is denoted by $n$, and the number of edges $|E|$ by $m$.

Erd\H{o}s-R\'{e}nyi (ER)~\cite{erdos1960evolution} have been widely used as fundamental models for network generation. However, these networks do not inherently account for spatial awareness or proximity information. One may question the relevance of incorporating \emph{proximity} in the context of disease spread. Notably, human populations exhibit natural clustering tendencies, which become evident when visualizing population distributions. This phenomenon is clearly illustrated in the population distributions of two major U.S. states: New York state and Texas (see Figure~\ref{fig:population-map}).

Proximity-based contact networks are crucial for modeling disease spread. Salathé et al.~\cite{salathe2010high} used high-resolution sensor data to map transmission, while Ferretti et al.~\cite{ferretti2020quantifying} highlighted digital contact tracing for epidemic control. Block et al.~\cite{block2020social} further demonstrated how network-based distancing strategies can flatten infection curves. These studies underscore the importance of spatially-aware models for effective disease mitigation.

\begin{figure}
    \centering
    \includegraphics[width=0.9\linewidth]{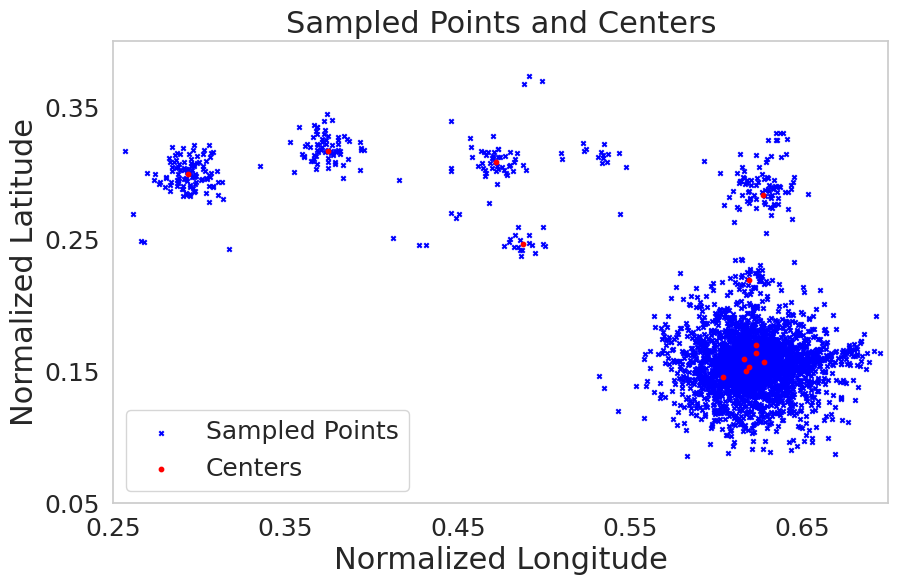}

    \includegraphics[width=0.9\linewidth]{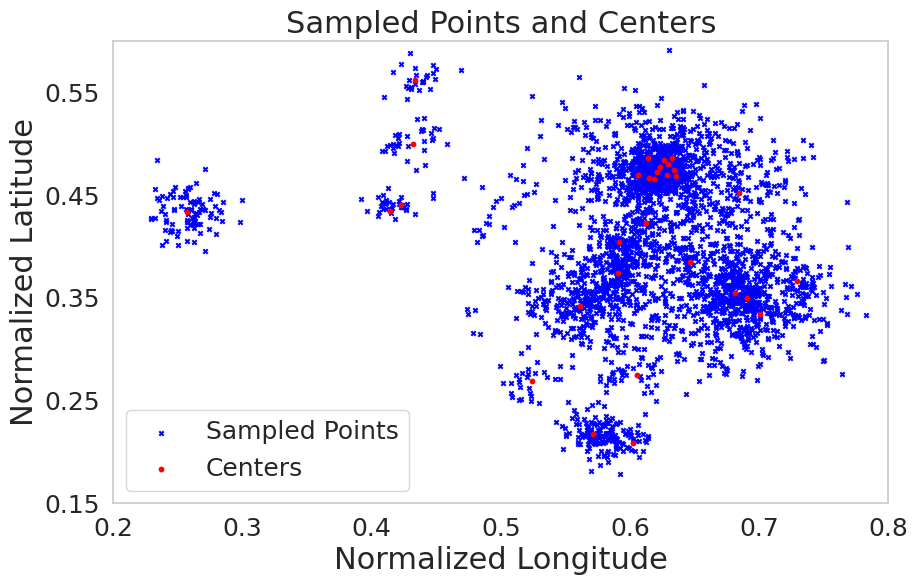}

    \caption{Population map of New York (top) and Texas (bottom)}
    \label{fig:population-map}

\end{figure}

We now define the two network models, (a) \emph{Independent Cascade} and (b) \emph{Linear Threshold}, in the context of Disease Spreading. 

\noindent\textbf{Independent Cascade (IC) Model}: 
The IC model is represented as a directed graph $G(V, E)$, where $V$ is the set of individuals, and $E$ is the set of directed edges. Each edge $(v_i,v_j)\in E$ is associated with a disease transmission probability $p_{ij} \in [0,1]$, representing the likelihood that a node $v_i$ will propagate the disease to a node $v_j$. The spreading process unfolds in discrete time steps, beginning with an initial set $\mathcal{I}\subset V$ of infected nodes. At each time step $t$, newly infected nodes from $t-1$ attempt to infect their uninfected neighbors with the assigned probabilities. Each spreading attempt is independent, and a node gets at most one chance to influence a neighbor. The process continues iteratively until no further activations occur, at which point the diffusion terminates. 

In our context, an individual $v_j$ may be more susceptible to contracting a disease from a close associate $v_i$, if the individual $v_j$ has bad hygiene practices. Hence, the set of edges is directed, with a probability associated with each edge. The probability $p_{ij}$ refers to the probability that the nodes= $v_j$ contracts a disease from the infected nodes $v_i$.

\noindent\textbf{Linear Threshold Model}: 
In the LT model, the strengths of interaction between people can be modeled as weights on the directed edges that connect the two individuals. This weight signifies the extent of transmission of the disease to an individual through these interactions. If an individual's body is unable to fight off the infection, then they get infected by the disease. 
Hence, we model the resistance of an individual as a threshold associated with that individual. Thus, if an individual $v_j$ closely interacts with many infected people beyond the body's resistance, they contract the disease. Moreover, they too begin to transmit the disease to others.

Note that because it is difficult to assess any individual's resistance to disease, we model it as a random number between $0$ and $1$. Additionally, as the edge weights indicate the amount of virus being delivered, it can add up to at most $1$, where $0$ represents no virus being delivered and $1$ is the highest amount that can be delivered.

To illustrate this network model, we present a running example.

\smallskip
\textbf{Running example} (\emph{Housing society}): Consider a bunch of $6$ individuals who live in the same housing society (or building). Some people go to the same vendor $v_2$ daily to buy vegetables. As $v_2$ is a young and healthy individual who often adheres to hygienic practices, their chances of contracting a disease from a single interaction are lower. But if $v_2$ contracts the disease, they may spread it to other individuals. This depends on the level of hygiene practices in their personal lives. Moreover, in this society, some individuals tend to buy newspapers from $v_5$. Friends like $v_3$ and $v_6$ meet over tea and snacks often. These interactions are modeled as probabilities (weights in the case of the LT model)
on edges and is illustrated in Figure~\ref{fig:sample-graph} (left).  

\begin{figure}
    \centering
    \resizebox{0.6\linewidth}{!}{\begin{tikzpicture}[
    triangle/.style = {fill=black!10, regular polygon, regular polygon sides=3 },
    node rotated/.style = {rotate=90},
  ]
  \tikzstyle{every node}=[font=\fontsize{15}{15}\selectfont]
\node[circle, fill={rgb,255:red,218;green,232;blue,252}, text=black, line width=0.3mm, draw={rgb,255:red,108;green,142;blue,191}, scale=1, minimum size=1cm](0) at (0, 0){$v_1$};
\node[circle, fill={rgb,255:red,218;green,232;blue,252}, text=black, line width=0.3mm, draw={rgb,255:red,108;green,142;blue,191}, scale=1, minimum size=1cm][right=2cm of 0](1) {$v_2$};
\node[circle, fill={rgb,255:red,218;green,232;blue,252}, text=black, line width=0.3mm, draw={rgb,255:red,108;green,142;blue,191}, scale=1, minimum size=1cm][right=1.5cm of 1](3) {$v_4$};
\node[circle, fill={rgb,255:red,218;green,232;blue,252}, text=black, line width=0.3mm, draw={rgb,255:red,108;green,142;blue,191}, scale=1, minimum size=1cm][right=1.5cm of 3](4) {$v_5$};
\node[circle, fill={rgb,255:red,218;green,232;blue,252}, text=black, line width=0.3mm, draw={rgb,255:red,108;green,142;blue,191}, scale=1, minimum size=1cm][below=1.5cm of 3](2) {$v_3$};
\node[circle, fill={rgb,255:red,218;green,232;blue,252}, text=black, line width=0.3mm, draw={rgb,255:red,108;green,142;blue,191}, scale=1, minimum size=1cm][below=1.5cm of 1](5) {$v_6$};
% \node[circle, fill={rgb,255:red,218;green,232;blue,252}, text=black,line width=0.3mm, draw={rgb,255:red,108;green,142;blue,191}, scale=1, minimum size=1cm][left=2cm of 5](6) {$v_7$};

\draw[->, line width=0.5mm, draw={rgb,255:red,108;green,142;blue,191}] (0) to[out=30,in=150]  node[above, xshift=0pt, yshift=-1pt]{0.2} (1);
\draw[->, line width=0.5mm, draw={rgb,255:red,108;green,142;blue,191}] (1) to[out=-150,in=-30]  node[below, xshift=0pt, yshift=-1pt]{0.1} (0);

\draw[->, line width=0.5mm, draw={rgb,255:red,108;green,142;blue,191}] (1) to[out=30,in=150]  node[above, xshift=0pt, yshift=-1pt]{0.8} (3);
\draw[->, line width=0.5mm, draw={rgb,255:red,108;green,142;blue,191}] (3) to[in=-30,out=-150]  node[below, xshift=0pt, yshift=-1pt]{0.1} (1);

\draw[->, line width=0.5mm, draw={rgb,255:red,108;green,142;blue,191}] (1) to[out=-120,in=120]  node[above, xshift=-13pt, yshift=-15pt] {0.3} (5);
\draw[->, line width=0.5mm, draw={rgb,255:red,108;green,142;blue,191}] (5) to[in=-60,out=60]  node[below, xshift=13pt, yshift=0pt] {0.1} (1);

% \draw[line width=0.5mm, draw={rgb,255:red,108;green,142;blue,191}] (2) to[out=90,in=-90]  node[below, xshift=-13pt, yshift=3pt] {0.1} (3);

\draw[->, line width=0.5mm, draw={rgb,255:red,108;green,142;blue,191}] (2) to[out=-150,in=-30]  node[below, xshift=0pt, yshift=-3pt] {0.5} (5);
\draw[->, line width=0.5mm, draw={rgb,255:red,108;green,142;blue,191}] (5) to[out=30,in=150]  node[below, xshift=0pt, yshift=-3pt] {0.5} (2);

\draw[->, line width=0.5mm, draw={rgb,255:red,108;green,142;blue,191}] (2) to[out=30,in=-110] node[below, xshift=0pt, yshift=0pt, rotate=45] {0.2} (4);
\draw[->, line width=0.5mm, draw={rgb,255:red,108;green,142;blue,191}] (4) to[out=-155,in=60] node[below, xshift=0pt, yshift=0pt, rotate=45] {0.2} (2);

\draw[->, line width=0.5mm, draw={rgb,255:red,108;green,142;blue,191}] (3) to[out=25,in=155]  node[above] {0.1} (4);
\draw[->, line width=0.5mm, draw={rgb,255:red,108;green,142;blue,191}] (4) to[out=-155,in=-25]  node[above] {0.1} (3);

% \draw[line width=0.5mm, draw={rgb,255:red,108;green,142;blue,191}] (2) to[out=45,in=-135] node[below, xshift=0pt, yshift=0pt, rotate=45] {0.2} (4);
% \draw[line width=0.5mm, draw={rgb,255:red,108;green,142;blue,191}] (3) to[out=0,in=180]  node[below] {0.1} (4);
% \draw[line width=0.5mm, draw={rgb,255:red,108;green,142;blue,191}] (5) to[out=180,in=0]  node[below, xshift=0pt, yshift=-1pt] {0.2} (6);

\end{tikzpicture}}

    \resizebox{0.6\linewidth}{!}{\begin{tikzpicture}[
    triangle/.style = {fill=black!10, regular polygon, regular polygon sides=3 },
    node rotated/.style = {rotate=90},
  ]
  \tikzstyle{every node}=[font=\fontsize{15}{15}\selectfont]
\node[circle, fill={rgb,255:red,252;green,218;blue,232}, text=black, line width=0.3mm, draw={rgb,255:red,191;green,108;blue,142}, scale=1, minimum size=1cm](0) at (0, 0){$v_1$};

\node[circle, fill={rgb,255:red,218;green,252;blue,232}, text=black, line width=0.3mm, draw={rgb,255:red,142;green,191;blue,108}, scale=1, minimum size=1cm][right=2cm of 0](1) {$v_2$};

\node[circle, fill={rgb,255:red,218;green,232;blue,252}, text=black, line width=0.3mm, draw={rgb,255:red,108;green,142;blue,191}, scale=1, minimum size=1cm][right=1.5cm of 1](3) {$v_4$};
\node[circle, fill={rgb,255:red,218;green,232;blue,252}, text=black, line width=0.3mm, draw={rgb,255:red,108;green,142;blue,191}, scale=1, minimum size=1cm][right=1.5cm of 3](4) {$v_5$};
\node[circle, fill={rgb,255:red,218;green,232;blue,252}, text=black, line width=0.3mm, draw={rgb,255:red,108;green,142;blue,191}, scale=1, minimum size=1cm][below=1.5cm of 3](2) {$v_3$};
\node[circle, fill={rgb,255:red,218;green,232;blue,252}, text=black, line width=0.3mm, draw={rgb,255:red,108;green,142;blue,191}, scale=1, minimum size=1cm][below=1.5cm of 1](5) {$v_6$};
% \node[circle, fill={rgb,255:red,218;green,232;blue,252}, text=black,line width=0.3mm, draw={rgb,255:red,108;green,142;blue,191}, scale=1, minimum size=1cm][left=2cm of 5](6) {$v_7$};

\draw[->, line width=0.5mm, draw={rgb,255:red,108;green,142;blue,191}] (0) to[out=30,in=150]  node[above, xshift=0pt, yshift=-1pt]{0.2} (1);
\draw[->, line width=0.5mm, draw={rgb,255:red,108;green,142;blue,191}] (1) to[out=-150,in=-30]  node[below, xshift=0pt, yshift=-1pt]{0.1} (0);

\draw[->, line width=0.5mm, draw={rgb,255:red,108;green,142;blue,191}] (1) to[out=30,in=150]  node[above, xshift=0pt, yshift=-1pt]{0.8} (3);
\draw[->, line width=0.5mm, draw={rgb,255:red,108;green,142;blue,191}] (3) to[in=-30,out=-150]  node[below, xshift=0pt, yshift=-1pt]{0.1} (1);

\draw[->, line width=0.5mm, draw={rgb,255:red,108;green,142;blue,191}] (1) to[out=-120,in=120]  node[above, xshift=-13pt, yshift=-15pt] {0.3} (5);
\draw[->, line width=0.5mm, draw={rgb,255:red,108;green,142;blue,191}] (5) to[in=-60,out=60]  node[below, xshift=13pt, yshift=0pt] {0.1} (1);

% \draw[line width=0.5mm, draw={rgb,255:red,108;green,142;blue,191}] (2) to[out=90,in=-90]  node[below, xshift=-13pt, yshift=3pt] {0.1} (3);

\draw[->, line width=0.5mm, draw={rgb,255:red,108;green,142;blue,191}] (2) to[out=-150,in=-30]  node[below, xshift=0pt, yshift=-3pt] {0.5} (5);
\draw[->, line width=0.5mm, draw={rgb,255:red,108;green,142;blue,191}] (5) to[out=30,in=150]  node[below, xshift=0pt, yshift=-3pt] {0.5} (2);

% \draw[line width=0.5mm, draw={rgb,255:red,108;green,142;blue,191}] (2) to[out=45,in=-135] node[below, xshift=0pt, yshift=0pt, rotate=45] {0.2} (4);
% \draw[line width=0.5mm, draw={rgb,255:red,108;green,142;blue,191}] (3) to[out=0,in=180]  node[below] {0.1} (4);
% \draw[line width=0.5mm, draw={rgb,255:red,108;green,142;blue,191}] (5) to[out=180,in=0]  node[below, xshift=0pt, yshift=-1pt] {0.2} (6);

\draw[->, line width=0.5mm, draw={rgb,255:red,108;green,142;blue,191}] (2) to[out=30,in=-110] node[below, xshift=0pt, yshift=0pt, rotate=45] {0.2} (4);
\draw[->, line width=0.5mm, draw={rgb,255:red,108;green,142;blue,191}] (4) to[out=-155,in=60] node[below, xshift=0pt, yshift=0pt, rotate=45] {0.2} (2);

\draw[->, line width=0.5mm, draw={rgb,255:red,108;green,142;blue,191}] (3) to[out=25,in=155]  node[above] {0.1} (4);
\draw[->, line width=0.5mm, draw={rgb,255:red,108;green,142;blue,191}] (4) to[out=-155,in=-25]  node[above] {0.1} (3);

\end{tikzpicture}}

    \caption{Example graph corresponding to a \emph{society} (top). Sample \emph{society} graph with infected node $\{v_1\}$ (bottom).}
    \label{fig:sample-graph}

\end{figure}
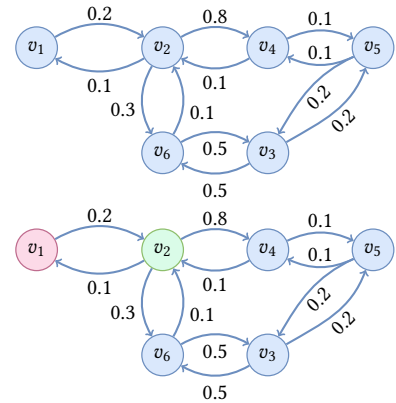

\smallskip\noindent
\textbf{Infected nodes}:
The governing body at some time-step identifies a set of individuals as infected nodes in the network. These nodes have the ability to spread the disease and create more infected nodes in the network. The initial set of nodes that are infected are denoted as $\mathcal{I}$.

\smallskip\noindent
\textbf{Vaccinated nodes}: Once the disease is detected in some part of the network, a few nodes in the network will be vaccinated. These vaccinated nodes cannot be infected and thus cannot transmit diseases. We denote the set of vaccinated nodes $\mathcal{S}$. 

In our running example, we mark the vaccinated nodes with the colour green and the infected nodes with the color red. For example, Figure~\ref{fig:sample-graph} (right) shows the vaccinated nodes $\mathcal{S}=\{v_2\}$ in the colour green and the initial set of infected nodes as $\mathcal{I}=\{v_1\}$ in red.

\smallskip\noindent
\textbf{Spreading mechanism}:

To describe the spreading mechanism in detail, we use similar notions as Kempe et al.~\cite{kempe2003maximizing}. If a node is infected in the current time step, we call such a node live. The disease tries to spread from a live node $v_i$ to its neighbor $v_j$ based on the strength of the interaction ($w_{ij}$ for the LT model and $p_{ij}$ for the IC model). An important distinction is that the transmission only succeeds if $v_j$ is not vaccinated. Moreover, in the case of the LT model, the threshold $\tau_j$ of $v_j$ should be exceeded by this change.

Kempe et al.~\cite{kempe2003maximizing} show that a deterministic sample of the LT model is equivalent to a network created by selecting a set of \emph{live edges}. These live edges are sampled from $G$, such that each node $v_j$ has at most one incoming edge. The probability of sampling an incoming edge from a node $v_i$ is equal to the edge weight. Node $v_j$ has no incoming edge with probability $1 - \sum_{j \in N(i)} w_{ij}$, where $w_{ij}$ is the weight of the directed edge connecting $v_i$ to $v_j$. 
The probability of sampling the edge $(v_i, v_j)$ is $w_{ij}$. Moreover, the sampled graph is an unweighted network with the property, \emph{if a node $v_j$ is reachable from any infected node, it contracts the disease} (see \cite{kempe2003maximizing} Claim 2.3).

A similar notion of live edges is present in the IC model. If a node is infected in the current time-step, we call such a node as current. The disease tries to spread from the current nodes $v_i$ to its neighbor $v_j$ for one time step with the probability $p_{ij}$. This process is similar to flipping a biased coin with a probability of $p_{ij}$. Instead of performing coin flips once a node is activated, one could perform $|E|$ coin flips independently, beforehand. This process determines the topology $X$ (similar to LT model) which will be used for disease spreading. 

The notable difference between our work and \cite{kempe2003maximizing} is the notion of vaccinated nodes, which cannot be infected (or \emph{influenced}). An alternate manner of thinking about disease-spreading is by considering the process as an influence-maximization over a modified network. The modified network is a vertex-induced sub-graph after deleting the vaccinated nodes ($V-\mathcal{S}$). Note that the infected set of nodes $\mathcal{I}$ is known beforehand, unlike influence maximisation.

We are now ready to introduce our problem.

\subsection{Problem definition}\label{subsec:problem-definition}

We assume that the initial set of infected nodes (say, $\mathcal{I}$) is known. To stop the spread of a disease, one idea is to identify the \emph{super spreaders} and vaccinate. An additional constraint is that the number of vaccines is limited due to cost and supply constraints. 
Hence, we need to carefully decide choose the nodes which need to be vaccinated. We formally define the problem below.

\begin{definition}[Limiting disease spread]
Given a population network (modeled as a graph $G(V, E)$), a set of infected nodes $\mathcal{I}\subseteq V$ and a budget of $k$ vaccines, find the $k$ nodes to be vaccinated in the graph which minimises the number of nodes that will be infected when the disease spreads in the network.     
\end{definition}

The optimization problem can be mathematically written as follows,

\begin{align}
    \underset{\mathcal{S}}{minimize}\ \ \sigma(\mathcal{S}, \mathcal{I})\ \ \ \  \ \ s.t.\ \ \ |\mathcal{S}| = k
\end{align}
where $\sigma(\mathcal{S}, \mathcal{I})$ represents the expected number of nodes infected through the disease-spreading process. 
\section{Hardness results}\label{sec:hardness-results}

\subsection{Submodularity and supermodularity results}\label{subsec:submodular}

Among the numerous papers that have looked at information-maximization "\emph{style}" problems~\cite{kempe2003maximizing}, one method that seems to provide a good approximation is the greedy technique. A common theoretical approach to prove the \emph{goodness} of the greedy technique~\cite{fisher1978analysis, nemhauser1978analysis} is to prove two important properties of the optimization function; (i) monotone and (ii) submodularity of the optimization function.
In this section, we analyze the optimization function of our problem as described in Section~\ref{subsec:problem-definition} for the two properties. 

\begin{lemma}[Monotone]\label{lemma:monotone}
    The optimization function of the \OurProb\ problem is monotone.
\end{lemma}

\begin{proof}
 
    Consider a topology $X$ that is determined by performing $|M|$ coin tosses beforehand. Additionally, the set of infected nodes $\mathcal{I}$ is given to us beforehand. Assume that the nodes of $S$ are selected and the number of nodes \emph{saved} are counted. Adding one more node to $\mathcal{S}$ prevents more nodes in the network from being \emph{saved} from the disease, as given below.
    \begin{align*}
        \sigma_X(\mathcal{S}\cup \{v_i\}, \mathcal{I}) \geq \sigma_X(\mathcal{S}, \mathcal{I})  
    \end{align*}
    As the monotone property is true for any given topology $X$, it extends to the general weighted version. Hence, proved.
\end{proof}

On the flip side, there exist topologies $X$ where the function $\sigma_X(\mathcal{S}, \mathcal{I})$ is neither submodular nor super-modular. We prove that in the following theorem.

\begin{lemma}\label{thm:negative-sub-super-modular}
    Given a graph $G(V, E)$ and a set of infected nodes $\mathcal{I}$, there exists a topology $X$ for which the function $\sigma_X(\mathcal{S}, \mathcal{I})$ is neither submodular nor super-modular.
\end{lemma}

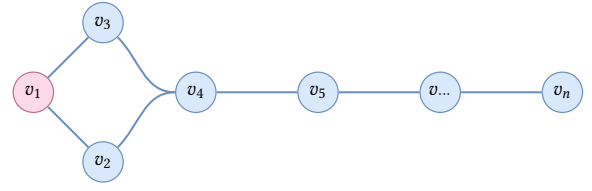
\begin{figure}[t]
	\centering
	\resizebox{.9\linewidth}{!}{\begin{tikzpicture}[
    triangle/.style = {fill=black!10, regular polygon, regular polygon sides=3 },
    node rotated/.style = {rotate=90},
  ]
  \tikzstyle{every node}=[font=\fontsize{15}{15}\selectfont]

\node[circle, fill={rgb,255:red,252;green,218;blue,232}, text=black, line width=0.3mm, draw={rgb,255:red,191;green,108;blue,142}, scale=1, minimum size=1cm](0) at (1, 0){$v_1$};

\node[circle, fill={rgb,255:red,218;green,232;blue,252}, text=black, line width=0.3mm, draw={rgb,255:red,108;green,142;blue,191}, scale=1, minimum size=1cm][below right=1cm and 1cm of 0](1) {$v_2$};

\node[circle, fill={rgb,255:red,218;green,232;blue,252}, text=black, line width=0.3mm, draw={rgb,255:red,108;green,142;blue,191}, scale=1, minimum size=1cm][above right=1cm and 1cm of 0](2) {$v_3$};

\node[circle, fill={rgb,255:red,218;green,232;blue,252}, text=black, line width=0.3mm, draw={rgb,255:red,108;green,142;blue,191}, scale=1, minimum size=1cm][right=3cm of 0](3) {$v_4$};

\node[circle, fill={rgb,255:red,218;green,232;blue,252}, text=black, line width=0.3mm, draw={rgb,255:red,108;green,142;blue,191}, scale=1, minimum size=1cm][right=2cm of 3](4) {$v_5$};

\node[circle, fill={rgb,255:red,218;green,232;blue,252}, text=black, line width=0.3mm, draw={rgb,255:red,108;green,142;blue,191}, scale=1, minimum size=1cm][right=2cm of 4](5) {$v_{\cdots}$};

\node[circle, fill={rgb,255:red,218;green,232;blue,252}, text=black, line width=0.3mm, draw={rgb,255:red,108;green,142;blue,191}, scale=1, minimum size=1cm][right=2cm of 5](6) {$v_n$};

\draw[line width=0.5mm, draw={rgb,255:red,108;green,142;blue,191}] (0) to[out=-45,in=135]  node[above, xshift=0pt, yshift=-1pt]{} (1);

\draw[line width=0.5mm, draw={rgb,255:red,108;green,142;blue,191}] (0) to[out=45,in=-135]  node[above, xshift=0pt, yshift=-1pt]{} (2);

\draw[line width=0.5mm, draw={rgb,255:red,108;green,142;blue,191}] (2) to[out=-45,in=-180]  node[above, xshift=0pt, yshift=-1pt]{} (3);

\draw[line width=0.5mm, draw={rgb,255:red,108;green,142;blue,191}] (1) to[out=45,in=-180]  node[above, xshift=0pt, yshift=-1pt]{} (3);

\draw[line width=0.5mm, draw={rgb,255:red,108;green,142;blue,191}] (3) to[out=0,in=-180]  node[above, xshift=0pt, yshift=-1pt]{} (4);

\draw[line width=0.5mm, draw={rgb,255:red,108;green,142;blue,191}] (4) to[out=0,in=-180]  node[above, xshift=0pt, yshift=-1pt]{} (5);

\draw[line width=0.5mm, draw={rgb,255:red,108;green,142;blue,191}] (5) to[out=0,in=-180]  node[above, xshift=0pt, yshift=-1pt]{} (6);

\end{tikzpicture}}
	\caption{Sample topology $X$.}
	\label{fig:sample_modular_graph}
\end{figure}

\begin{proof}
    For this proof, we construct a topology $X$ for which we prove the result. Consider the topology in Figure~\ref{fig:sample_modular_graph}. 
    
    We first prove that for a selected set of nodes $A\subset B$, addition of a node $v_i$ to $A$ provides more gain to $A$ than adding the same node to $B$. Consider the sets $A=\{v_3\}$ and $B=\{v_3, v_4, v_5\}$. Adding the node $v_2$ to $A$ would save the remaining $n-3$ nodes in the network. Hence, the gain is 
    \begin{align*}
    \sigma_X\{A\cup \{v_2\}, \mathcal{I}=\{v_1\})-\sigma_X(A, \mathcal{I}=\{v_1\})=n-2    
    \end{align*}
    But adding $v_2$ to $B$, $\sigma_X$ would gain only $1$ node. Hence, for the chosen set $A\subset B$, the gain to the subset $A$ is greater than that of the superset $B$.

    We now prove the exact opposite inequality using a different set of nodes. Consider the set of nodes $A=\emptyset$ and $B=\{v_3\}$. While $\sigma_X$ gains $1$ node upon adding the node $v_2$ to $A$, adding the same node $v_2$ to $B$ would save all the other nodes in the network (\emph{i.e.} gain of $n-2$ nodes). Thus, the gain to the superset is greater than that of the subset.

    Based on the above two choices of $A$ and $B$ ($A\subseteq B$) and a given set of infected nodes $\mathcal{I}$, we can conclude that the function $\sigma_X$ is neither submodular nor super-modular.
\end{proof}

\section{Algorithmic results}\label{sec:algorithms}

We first present the computational challenge that is inherent in the problem. Then we describe a sampling solution approach that alleviates the computational bottleneck at the cost of efficacy. We then present a broad range of approaches to solve the problem which utilizes the sampling-based framework. To make the time complexity bounds applicable to both Linear Threshold and Independent Cascade models, we maintain the number of edges in each sample to be $\mathcal{O}(m)$, unlike the Linear Threshold bound of $\mathcal{O}(n)$.

\subsection{Sampling strategy}\label{subsec:sampling-approach}

Consider a graph with $n$ nodes and $m$ edges. Based on the weights on each edge and threshold of each vertex, the propagation could infect different nodes. To further analyse our scenario, we describe a technique of sampling unweighted topologies and performing our computation on these topologies.

For the independent cascade model, each edge has a probability associated with it which indicates the probability of a node infecting its neighbor. To simulate a run of the propagation process, a coin flip is performed at each time step to check whether an edge should be retained or discarded (whether a node infects its neighbor or not). We can instead perform all $m$ coin flips at the start to obtain a topology, which is a directed graph without any weights on its edges. For the linear threshold model, each edge has a score/weight associated with it, which indicates the extent of disease-propagation across the edge, and each node has a threshold that indicates its susceptibility to the disease. We can perform $n$ coin flips at the start to obtain the thresholds for all nodes, but instead use the alternate equivalent formulation described by Kempe et al.~\cite{kempe2003maximizing} of sampling, for each node, at most one incoming edge with a probability equal to its edge weight, resulting in an undirected graph with unweighted edges, similar to the independent cascade case.

A naive approach to enumerate the set of all possible topologies and compute the set of $k$ nodes that minimise the spread of the disease. While this approach will be able to compute the optimal set of $k$ nodes, enumerating the set of all topologies is exponential in $m$ ($\mathcal{O}({2^m})$ topologies). Instead, we propose a sampling based approach to overcome the computational bottleneck. Our approach is based on randomly sampling $s$ topologies from the set of all topologies and solving the problem of selecting $k$ nodes such that the sum of the nodes infected across all topologies is minimized.

Upon sampling $s$ topologies and minimizing over these samples, we estimate the number of nodes saved. This provides us with an estimate on the optimal value of $\sigma$. From Bayraksan et al.~\cite{bayraksan2006assessing} we know this estimate has a negative bias, i.e. it always estimates a score that is lower than the optimal value. Hence, we use the estimated score of $\sigma$ over $s$ samples as a solution to our problem. 

We propose different solutions to the sub-problem of finding the $k$ nodes which minimize the number of infected nodes over the $s$ sampled topologies next.

\subsection{Binary Linear Program and relaxation}\label{subsec:ilpr-relaxation}

Our first approach is a linear program with binary variables (BLP) to select the set of nodes $\mathcal{S}$ for vaccination. The BLP is formulated based on the weights on the edges $W$. Based on $W$ around $2^{|E|}$ topologies could be possible. The exact computation of $\mathcal{S}$ using a Binary Linear Program is formulated over the set of all possible topologies $\mathcal{T}$.

\begin{align}\label{eqn:exponetital-ilp}
    \begin{split}
        minimize &\sum_{G(V,E)\in \mathcal{T}}\ \ \sum_{i\in V}\ \ \mu_G\ x_{i,G}\\
        &  i\in V,\ G(V,E)\in \mathcal{T}\ \  \forall\ {j=\eta(G,i)}\ \ \  x_{i,G} \geq x_{j, G} - I_i \\
        & i \in V,\ G(V,E)\in \mathcal{T} \ \ \ \ \ I_j \in \{0, 1\} \ \ , 1 \geq x_{i, G} \geq 0\\
        & \sum_{i\in V\backslash \mathcal{I}}I_i \leq k
    \end{split}
\end{align}

where $\mu_G$ represents the probability of topology $G(V,E)$ being generated and $x_{i,G}$ accounts for the infection of $v_i$ in $G(V,E)$.

The most expensive part of this formulation is the enumeration of the $2^{|E|}$ topologies ($\mathcal{T}$). We reduce this complexity using sampling. But it comes at the cost of the quality of the result. We describe this sampling-based approach now.

\smallskip\noindent
\textbf{Sampling-based BLP}: Instead of enumerating all the topologies  $\mathcal{T}$, our idea relies on sampling a subset of them and computing the nodes for vaccination based on these samples. This will provide us an estimate of the total error. The formulation is similar to that of Equation~\ref{eqn:exponetital-ilp}, except for the enumeration of the topologies. We replace the set of all topologies $\mathcal{T}$ with the sampled topologies $\mathcal{U}$. The final formulation is shown below.

\begin{align}\label{eqn:sampled-ilp}
    \begin{split}
        minimize &\sum_{G(V,E)\in \mathcal{U}}\ \ \sum_{i\in V}\ \ \mu_G\ x_{i,G}\\
        &  i\in V,\ G(V,E)\in \mathcal{U}\ \  \forall\ {j=\eta(G,i)}\ \ \  x_{i,G} \geq x_{j, G} - I_i \\
        & i \in V,\ G(V,E)\in \mathcal{U} \ \ \ \ \ I_j \in \{0, 1\} \ \ , 1 \geq x_{i, G} \geq 0 \\
        & \sum_{j\in V\backslash \mathcal{I}} I_j\leq k
    \end{split}
\end{align}

The above approach provides us the optimal set of $k$ nodes that need to be vaccinated, given the $|\mathcal{U}|=s$ topologies. However, this procedure may not be time effective. Thus, we propose a relaxation of the Binary constraints to take non-integer values. To be more specific, the constraints on $I_j$ are relaxed such that $I_j$ can take values between $0$ and $1$ (i.e. $I_j \in [0, 1]$). The relaxed LP is presented below.

\begin{align}\label{eqn:sampled-lp-relaxed}
    \begin{split}
        minimize &\sum_{G(V,E)\in \mathcal{U}}\ \ \sum_{i\in V}\ \ \mu_G\ x_{i,G}\\
        &  i\in V,\ G(V,E)\in \mathcal{U}\ \  \forall\ {j=\eta(G,i)}\ \ \  x_{i,G} \geq x_{j, G} - I_i \\
        & i \in V,\ G(V,E)\in \mathcal{U} \ \ \ \ \ 0 \leq I_j \leq 1 \ \ , 1 \geq x_{i, G} \geq 0 \\
        & \sum_{j\in V\backslash \mathcal{I}} I_j\leq k
    \end{split}
\end{align}

While the relaxation does not provide an exact result, it provides us a lower bound on the number of infected nodes. But as the nodes that would be vaccinated can take fractional values, we would need a rounding procedure to obtain a valid set of nodes to vaccinate. To obtain valid solutions, we describe two rounding procedures, top-$k$ rounding (TKR) and an iterative rounding procedure (IRP). 

\smallskip\noindent
\textbf{Top-$k$ Rounding procedure} (\texttt{TKR}): Our first rounding approach is a deterministic procedure that considers only the value for the vaccination variables ($I_j$). As the vaccines that have received a high score for vaccination may contribute significantly for the spread of the disease, we use the vaccination score as a key indicator. Among all the nodes that have received a non-zero score for vaccination, we pick those nodes that have the top $k$ values.

\smallskip\noindent
\textbf{Iterative Rounding procedure} (\texttt{IRP}): An alternative rounding procedure is an iterative one. Again, we use the score of $I_j$ as an indicator for the nodes to vaccinate. Instead of $k$ vaccinating nodes at once, we propose an iterative procedure that vaccinates one node at a time. We start with the relaxed LP with no node being vaccinated. In the first iteration of the relaxed LP, we select the node with the highest score $I_j$ to be vaccinated. We vaccinate the chosen node $c$ in the next iteration by setting the score of $I_c$ to $1$ by adding the constraint $I_c=1$ to the LP. We then rerun the LP and select the next node with the highest vaccination score. This process is continued iteratively until $k$ nodes are selected.

We perform extensive experimental evaluation (Section~\ref{sec:experiments}) to assess the performance of the relaxed LP with \texttt{TKR} and \texttt{IRP} rounding procedures.
%
%
% \textcolor{red}{If we do not implement this, we should skip the above paragraph before submission.}
%
The runtime analysis of the relaxed Linear Program with the two rounding procedures - \texttt{TKR} and \texttt{IRP} are presented in Theorem~\ref{thm:runtime-lp-relaxation-tkr} and Theorem~\ref{thm:runtime-lp-relaxation-irp}, respectively.

\begin{theorem}\label{thm:runtime-lp-relaxation-tkr}
    The relaxed LP with \texttt{TKR} rounding procedure consumes a total of $\mathcal{O}((m + n)^{1.5}\  m s^{2.5}\ L)$ amount of time, where $s$ is the number of topologies sampled and $L$ is the number of bits needed to represent the input.
\end{theorem}

\begin{proof}
    Consider a graph $G$ with $n$ vertices and $m$ edges. In the LP given in Equation~\ref{eqn:sampled-lp-relaxed}, there are a total of $n s$ variables that correspond to $x_{i,G}$ and $n$ variables that correspond to $I_j$. Thus, the total number of variables in the LP is $\mathcal{O}(n s)$. 
    
    The number of edges in the initial graph $G$ is $m$, and no more than $o(m)$ could be present in each sampled graph. Using this property, there are at most $m s$ equations in the graph that correspond to $x_{i, G}$.     Additionally, we add $2n$ constraints to restrict $I_j$ between $0$ and $1$. Hence, there are a total of $2n + m s$ number of equations in the LP.

    The total running time of an LP with $N$ equations and $D$ variables is $\mathcal{O}(N+D)^{1.5}NL$ (see Vaidya~\cite{vaidya1989speeding}) where $L$ is the number of bits needed to represent the input. Using this result, the running time of the $LP$ is $(m s+2 n + n s)^{1.5}\ m sL$. Simplifying, we get the overall complexity to be $\mathcal{O}((m + n)^{1.5}\ m s^{2.5}\  L)$. 
    
    The rounding procedure sorts the values of $I_j$ from the solution of the LP. Post that, the top $k$ elements are selected for vaccination. The rounding procedure consumes a total of $\mathcal{O}(n\log{(n}))$ time.

    The overall running time is dominated by the LP and thus, the total runtime complexity is $\mathcal{O}((m + n)^{1.5}\  m s^{2.5}\ L)$.
\end{proof}

\begin{theorem}\label{thm:runtime-lp-relaxation-irp}
    The relaxed LP with \texttt{IRP} rounding procedure consumes a total of $\mathcal{O}(k(m + n)^{1.5}\  m s^{2.5}\ L)$ amount of time, where $s$ is the number of topologies sampled and $L$ is the number of bits needed to represent the input.
\end{theorem}
\begin{proof} 
    (\emph{Sketch}) We use the runtime result of the LP from the proof of Theorem~\ref{thm:runtime-lp-relaxation-irp}. In the iterative rounding procedure, the LP is run $k$ times to identify the $k$ vertices that will be vaccinated. Thus, the overall running time is $k$ times that of \texttt{TKR}. Hence, the runtime is $\mathcal{O}(k(m + n)^{1.5}\  m s^{2.5}\ L)$.
\end{proof}

Note that in the above two theorems, we could compute the expected running time of our approach based on the expected number of constraints that would be added to the LP. As each edge (between $v_i$ and $v_j$) has a probability $p$ to be sampled in a topology, we could estimate the number of constraints by computing the expected number of edges present in the $s$. For each edge, the contribution towards the constraints in one topology is $m'=\sum_{(v_i, v_j)\in E}\ p_{ij}$. Thus, we can compute the expected running time of our approaches by replacing $m$ with $m'$ in the above formulae.

\subsection{Greedy Strategy}\label{subsec:greedy}

\begin{algorithm}[t]
\small{
    \caption{Greedy}\label{algo:greedy}
    \hspace*{\algorithmicindent} \textbf{Input} : Graph $G(V, E)$, topologies $T$, budget $k$, infected $\mathcal{I}$ \ \ \ \ \ \ \ \ \ \ \\
    \hspace*{\algorithmicindent} \textbf{Output} : Set of vertices $S$ to be vaccinated \ \ \ \ \ \ \ \ \ \ \ \ \ \ \ \ \ \ \ \ \ \ \ \ \ \ \ \ \ \ \ \ \ \ \ \ \ \ 
    \begin{algorithmic}[1]
        \State $S \leftarrow \emptyset$
        \For{$i \leftarrow\ 0$ to $k$}
            \For{$v \in V$}
                \State $saved[v]  \leftarrow 0$
                \For{$X \in T$}
                    \State $saved[v] \leftarrow saved[v] + \sigma_X(S\cup \{v\}, \mathcal{I})-\sigma_X(S, \mathcal{I})$
                \EndFor
            \EndFor
            \State Search the vertex $v$ which has the highest score in $saved$
            \State $S\leftarrow S\cup \{v\}$
        \EndFor
        \Return $S$
    \end{algorithmic}
}
\end{algorithm}

We now propose a greedy heuristic that is inspired by the nature of the optimization. The optimization problem of finding $k$ vertices that need to be vaccinated can be performed in $k$ steps. 

Our approach is a simple greedy approach where we start the iterative process with an empty set (i.e. $\mathcal{I}=\emptyset$). The different topologies are sampled from the graph $G(V, E)$ to estimate $\mathcal{I}$. In the first iteration, we add one vertex to the set such that it maximises the sum of vertices saved across all $s$ topologies. If two vertices pose a similar potential to save vertices, we break ties arbitrarily. The iterative process is continued until we add $k$ vertices to the graph. The pseudocode of our approach is presented in Algorithm~\ref{algo:greedy}.

As the optimization function is neither sub-modular nor super-modular (see Lemma~\ref{thm:negative-sub-super-modular}), the well known result from Nemhauser et. al~\cite{nemhauser1978analysis} is not applicable. Nevertheless, we show that Greedy works well in practice with extensive experimental evaluations. 

The runtime analysis of our approach is analysed in Theorem~\ref{thm:greedy-runtime}.

\begin{theorem}\label{thm:greedy-runtime}
    Given a set $s$ of topologies, Greedy Algorithm~\ref{algo:greedy} runs in time $\mathcal{O}(kns(n+m))$ time.
\end{theorem}

\begin{proof}
    (\emph{Sketch}) We analyse the greedy approach over one topology and then extend our analysis to the different topologies. 
    
    Consider the selection of the first vertex for vaccination. As any of the $n$ vertices could potentially be a candidate, we need to compute its contribution towards $\sigma$. Thus, one simple approach is to run a BFS or DFS to compute the spread of the disease when vertex $v_i$ is vaccinated. As we measure the potential of each vertex $v_i$ in $\mathcal{O}(n+m)$ time, the overall time needed to compute the first vertex to be vaccinated is $\mathcal{O}(n(n+m))$.

    As the same procedure needs to be repeated $k$ times across $s$ samples, the overall time is $\mathcal{O}(kns(n+m))$.
\end{proof}

\subsection{Local Search based approach}\label{subsec:local-search}

\begin{algorithm}[!t]
\small{
    \caption{Local Search}\label{algo:local-search}
    \hspace*{\algorithmicindent} \textbf{Input} : Graph $G(V, E)$, topologies $T$, vaccinated $S$, infected $\mathcal{I}$\ \ \ \ \ \ \ \ \ \ \ \ \ \ \ \ \ \\
    \hspace*{\algorithmicindent} \textbf{Output} : Set of vertices $S'$ to be vaccinated \ \ \ \ \ \ \ \ \ \ \ \ \ \ \ \ \ \ \ \ \ \ \ \ \ \ \ \ \ \ \ \ \ \ \ \ \ \ \ \ \ \ \ \ \ \ \ \ \ \ \ \ 
    \begin{algorithmic}[1]
        \State $S' \leftarrow S$; \ \ $converged\leftarrow$ \texttt{False}
        \While{not $converged$}
            \State $converged\leftarrow$ \texttt{True};  $temp \leftarrow S'$
            \For{$v \in S'$}
                \For{$v' \in \eta(v)$} \Comment{$\eta(v)$ gives the neighbors of $v$}
                    \If{$\sum_{X\in T} \sigma_X(temp, \mathcal{I}) < \sum_{X\in T} \sigma_X((S'\cup\{v'\})-\{v\}, \mathcal{I})$}
                        \State $converged\leftarrow$ \texttt{False}
                        \State $temp \leftarrow (S'\cup\{v'\})-\{v\}$
                    \EndIf
                \EndFor
            \EndFor
            \State $S'\leftarrow temp$
        \EndWhile
        \Return $S'$
    \end{algorithmic}
}
\end{algorithm}

A different heuristic that we try as a practical approach is based on the locality. Consider a set of $k$ vertices that form a feasible vaccination candidate. We propose a local search approach to modify the feasible solution to a better one in the locality. We define the notion of locality in definition~\ref{def:vaccination-locality} and then propose the heuristic.

\begin{definition}[neighbor]\label{def:vaccination-locality}
    Given a set $X$ of $k$ vertices that form a potential candidate for vaccination, a set of $k$ vertices $X'$ is termed as its neighbor if the two sets $X$, $X'$ differ by one vertex. Additionally, the vertex $v_i$ from $X$ that is replaced by $v_j$ in $X'$ should be neighbors. More concretely, there needs to be an edge between $v_j$ and $v_i$.   
\end{definition}

Based on the definition of a neighbor, we explore the search space using a local search approach. Our approach begins with a set of vertices $X$. These set of vertices could be chosen arbitrarily or from the output of another algorithm like the greedy approach. Note that $X$ and its neighbors are potential solutions to our optimization problem. Thus, our goal of the local search is to improve the optimization score in each iteration until convergence.

In each iteration, we find that neighbor $X'$ of $X$ which improves the optimization function the most. After updating our potential solution (initially $X$) to $X'$, our search continues in the neighborhood of $X'$. The local search procedure concludes the search when there are no neighbors which improves the optimization score. The pseudocode for the LS approach is presented in Algorithm~\ref{algo:local-search}.

We analyse the running time of each iteration of the local search approach in Theorem~\ref{thm:local-search-runtime}.

\begin{theorem}\label{thm:local-search-runtime}
    Given a set $X$ of vertices to be vaccinated, one step of the local search approach runs in an expected time of $\mathcal{O}(ks\frac{m}{n}(n+m))$. 
\end{theorem}
\begin{proof}
    We first compute the number of neighbors of $X$ and then for each neighbor compute the number of vertices that would be saved. Consider a set $X$ of $k$ vertices. In the graph $G(V, E)$, each vertex that belongs to $X$ in expectation may have $\mathcal{O}(\frac{m}{n})$ edges incident to it. As each of the $k$ vertices in $X$ could be replaced with a different vertex connected by an edge, the neighborhood could be at most $\mathcal{O}(k\frac{m}{n})$. Additionally, for each of these neighbors, a Breadth First Search (BFS), or Depth First Search (DFS) approach needs to be run over each topology to compute the extent of transmission. Hence, our approach consumes $\mathcal{O}(s(n+m))$ time for each neighbor. Thus, the expected running time is $\mathcal{O}(ks\frac{m}{n}(n+m))$.
\end{proof}

\section{Experiments}\label{sec:experiments}

\subsection{Experimental Setup}\label{subsec:experimental-setup}

We use two real-world datasets for these experiments. We generate synthetic datasets through popular graph generation techniques like the Waxman model~\cite{waxman1988routing} and the Erdős–Rényi model~\cite{erdos1959random}.  

We perform experiments on both the Linear Threshold and the Independent Cascade model. Note that the main difference between the two models is in the way weights or probabilities are assigned to the edges. Therefore, we follow the same technique to generate the networks and assign edge weights (or probabilities) in the following manner for the two models:
\begin{enumerate}[leftmargin=*]
    \item \textit{Linear Threshold model}: The edges are assigned weights randomly from a uniform distribution. To obey the constraints of the model, each edge weight is then rescaled so that the total weight of the incoming edges to each node is less than $1$.

    \item \textit{Independent Cascade model}: The edges are assigned probabilities randomly from a uniform distribution between $0$ and $1$. There is no correlation between any two edges' probabilities.
\end{enumerate}

\smallskip
\noindent\textbf{Datasets}:
\begin{itemize}[leftmargin=*]

    \item \textbf{Erdős–Rényi datasets}: We use the Erdős–Rényi graph generation technique \cite{erdos1959random} to generate random graphs with a fixed number of nodes. We assign edge weights/probabilities for each model using the methods described above.

    \item \textbf{Gaussian Waxman datasets}: To emulate a set of cities where populations are generally clustered around city centers, we randomly select some points from a 2-dimensional box to serve as the city centers. These centers act as the centers of a Gaussian distribution from which nodes for the graph are sampled. The variance for each Gaussian is sampled from a uniform random distribution and each city's population is proportional to its variance. This is done to model that cities with a larger spread generally have a higher population. After the specified number of nodes is sampled, we use the Waxman method \cite{waxman1988routing} to generate edges, in which an edge is created between two nodes with a probability $\alpha e^{-d / \beta L}$, where $d$ is the distance between the two nodes, $L$ is the maximum distance between any two nodes and $\alpha$ and $\beta$ are parameters. After generating the unweighted graph in this manner, edges are assigned weights using the same method as the Erdős–Rényi datasets.
    \item \textbf{New York and Texas datasets}: We obtain the datasets of the city-wise population and density of New York and Texas states from the dataset provided by \href{https://simplemaps.com/data/us-cities}{simplemaps.com}. The city-center coordinates, provided in the dataset, are then mapped to a rectangle of fixed length. The population sampling and graph generation are done using these coordinates as the city centers for the Gaussian Waxman method. We use the population density data along with the population data to get the variance of each city's Gaussian as a factor of $\sqrt{poulation / density}$. Since it is infeasible to have one node in the network for each person, the population is scaled down by some factor $f$, so that one node represents $f$ people. The edge sampling and weight assignments are done using the method used in the Gaussian Waxman networks.
\end{itemize}

\subsection{Algorithms and Implementation}

The graph generation and edge-weight assignment for both models are done using Python scripts.
The greedy, local search and hill-climbing algorithms are implemented in C++. The LP and ILP are implemented in Python using the \href{https://pypi.org/project/gurobipy/}{\texttt{gurobipy}}~\cite{gurobiCite} package. The programs are run on a machine with the 12th Gen Intel(R) Core(TM) i7-12700K processor with 12 cores and 20 CPUs. The code can be accessed at \href{https://anonymous.4open.science/r/DiseaseSpreading-4228}{this link}.
\par
For each network, the topology samples are generated once and the same samples are used across all programs. Doing so allows us to compare the performance of each method for a certain network and set of samples. The linear search and hill-climbing algorithms use the greedy solution as the initial solution and improve on it.

As a baseline, we implement a Hill Climbing (HC) based approach. Similar to our Local Search, HC starts from a set of vertices to be vaccinated. In each iteration, HC tries to exchange one vertex of the solution set by a different one. Among all the options that HC has for an exchange, the one that improves the optimization function the most is chosen for the next iteration. In any iteration, if none of the exchanges results in a better/improved solution, then the approach has converged, and we report the solution that has been found. Note that the HC approach finds the local optimum under the exchange operation.

\subsection{Experimental Results for Linear Threshold}

\subsubsection{Quality results for Gaussian Waxman graphs}\label{subsubsection:quality-waxman}

In this experiment, we measure the quality of the results for the different methods proposed in this paper. We generated graphs with 64, 128, 256, and 512 vertices using the Gaussian Waxman generation method for 5 centers with 10\% of the nodes initially infected which were chosen randomly, and the budget of vaccines being 10\% of the total population. For each of these graphs, we sampled 50 topologies.

We measure the quality of the different algorithms and report in Table~\ref{table:avg-nodes-saved-Gaussian-Waxman}. Note that the BLP finds the global optimum for the given set of topologies. Thus, in this experiment, we compare the quality against that of BLP. As mentioned before, both Local Search (LS) and Hill Climbing (HC) were initialized with the output of the greedy algorithm. Thus, one could see that they always outperform the greedy approach. But, as can be seen in these experiments, Greedy finds a solution very near to the optimal solution in many of these settings. When the number of nodes increases, the number of nodes saved by Greedy slightly drops, but the solution is still very near (off by around 2 vertices) to that of the BLP algorithm.

\texttt{LP} with \texttt{TKR} also performs reasonably well in these experiments. As the number of nodes increases, Greedy performs marginally better than \texttt{LP}-\texttt{TKR}. Both HC and LS improve the solution of Greedy by slight margins. Neither succeeds in finding the global optimum for the larger settings. Additionally, HC performs slightly better than LS, due to the larger search space of HC.  \texttt{LP} with \texttt{IRP} performs well compared \texttt{TKR} rounding and Greedy procedure. In some cases, it performs slightly better than LS and HC.

\begin{table}[t]
    \centering
    \resizebox{\linewidth}{!}{%
    \begin{tabular}{|c|c|c|c|c|c|c|}
        \hline
        \# Nodes & Greedy & LS & HC & BLP & LP - \texttt{TKR} & LP - \texttt{IRP} \\
        \hline
        64  & 54.34 & 54.34 & 54.34 & \textbf{54.34} & 54.34 & 54.34\\
        128   & 97.94 & 98.16 & 98.16 & \textbf{98.26} & 98.26 & 98.26 \\
        256  & 189.78  & 190.22 & 190.22 & \textbf{190.4} & 188.98 & 189.88\\
        512 & 369.12 & 369.7 & 369.8 & \textbf{370.28} & 365.18 & 369.62\\
        \hline
    \end{tabular}%
    }
    \caption{Average number of nodes saved by different algorithms for the Gaussian Waxman graph}\label{table:avg-nodes-saved-Gaussian-Waxman}
  
\end{table}

\subsubsection{Quality results for Erdős-Renyi graphs}\label{subsubsection:quality-renyi-erdos}

We perform a similar experiment as Section~\ref{subsubsection:quality-waxman} but with Erdős–Rényi graphs. For this experiment, we generated Erdős–Rényi graphs with 128, 256, and 512 nodes with around 10\% of the nodes initially infected. The infected nodes were chosen randomly, and the budget of vaccines was set to 10\% of the total population. Similar to our previous experiment, for each of the graphs, we sampled 50 topologies.

The results for this experiment are presented in Table~\ref{table:avg-nodes-saved-ER}. Similar to our previous experiment, the quality of Greedy is comparable to the optimal solution of BLP. LP relaxation with TKR performs close to Greedy but slightly worse compared to Greedy. LS and HC improve the solution of Greedy but very small (almost negligible) margins. \texttt{LP}  with \texttt{IRP} works better than the \texttt{TKR} procedure. But the Greedy performs better than \texttt{LP}-\texttt{IRP} for smaller graph sizes. At 512 nodes, \texttt{LP}-\texttt{IRP} performs marginally better than Greedy.

\begin{table}[t]
    \centering
    \resizebox{\linewidth}{!}{%
    \begin{tabular}{|c|c|c|c|c|c|c|}
        \hline
        \# Nodes & Greedy & LS & HC & ILP & LP - \texttt{TKR} & LP - \texttt{IRP}\\
        \hline
        128 & 94.3 & 95.02 & 95.26 & \textbf{95.26} & 94.96 & 95.2\\
        256  & 184.34  & 184.34 & 184.34 & \textbf{184.44} & 182.26 & 184.02\\
        512 & 361.4 & 362.62 & 362.92 & \textbf{363.12} & 358.12 & 362.14\\
        \hline
    \end{tabular}%
    }
    \caption{Average number of nodes saved by different algorithms ER graphs}\label{table:avg-nodes-saved-ER}

\end{table}

\subsubsection{Varying the \% budget of vaccines} \label{subsec:varying-budget-vaccines}
We explore the effect of varying the number of vaccines and observing the number of vertices that are saved from the infection. For the dataset, we generated Gaussian Waxman graphs with 512 nodes. We set 10\% of the vertices to be initially infected and fixed the number of topologies sampled to 50.

We vary the number of vaccinated vertices to 5\% (25 vertices), 10\% (51 vertices), 20\% (102 vertices), 30\% (138 vertices), 40\% (204 vertices), and 50\% (256 vertices) of the total population.

Our results are presented in Tables~\ref{table:nodes-saved-budget-vaccines-varied} and ~\ref{table:nodes-saved-budget-vaccines-varied-2}. For the higher vaccine budgets (30\%, 40\%, and 50\%), we skip the hill-climbing algorithm and ILP since they consume too much time. Moreover, hill-climbing provides negligible improvement over the greedy solution. 

An interesting observation is that when the number of vaccines increases from 5\% to 10\% (by around 5\% which is roughly 25 vaccines) the total number of people saved by the BLP increases by around 40, which is a significant increase in the number of people saved. But as the number of vaccines increases by another 10\% we do not see a significant jump in the number of people saved. By increasing the budget from 40\% to 50\%, we see an even smaller increase in the number of people saved for the other approaches. This shows a similar behaviour to the \emph{law of diminishing returns}, even though the function is not sub-modular for the general class of graphs. 

\begin{table}[t]
    \centering
    \resizebox{0.95\linewidth}{!}{%
    \begin{tabular}{|c|c|c|c|c|c|c|}
        \hline
        Budget \% & Greedy & LS & HC & ILP & LP - \texttt{TKR} & LP - \texttt{IRP}\\
        \hline
        5 & 327.26 & 327.26 & 328.02 & \textbf{328.32} & 326.6 & 328.32\\
        10  & 369.12 & 369.7 & 369.8 & \textbf{370.28} & 365.18 & 369.62\\
        20 & 409.96 & 410.38 & 410.54 & \textbf{411.08} & 408.32 & 410.34\\
        \hline
    \end{tabular}%
    }
    \caption{Average number of nodes saved when varying the vaccination budgets}\label{table:nodes-saved-budget-vaccines-varied}
 
\end{table}

\begin{table}[t]
    \centering
    \resizebox{0.8\linewidth}{!}{%
    \begin{tabular}{|c|c|c|c|c|}
        \hline
        Budget \% & Greedy & LS & LP - \texttt{TKR} & LP - \texttt{IRP}\\
        \hline
      
        30 & 429.56 & 430.74  & 429.7 & 430.96\\
        40 & 443.02 & 444.3  & 444.26 & 444.34\\
        50 & 451.02 & 451.52  & 451.58 & 451.58\\
        \hline
    \end{tabular}%
    }
    \caption{Average number of nodes saved when varying the vaccination budgets for higher values}\label{table:nodes-saved-budget-vaccines-varied-2}
    
\end{table}

\subsubsection{Runtime results for Gaussian Waxman graphs}\label{subsec:runtime-waxman}

In this experiment, we measure the runtime of the different algorithms for the Gaussian Waxman model. In this experiment, the number of vertices is varied with around 10\% of the vertices initially infected with the disease with a budget of 10\% vaccines. When we consider the time consumed by the different approaches, we find a different picture from the quality experiments seen so far. 

The results of this experiment are presented in Figure~\ref{fig:time-plot-Waxman}. Surprisingly, the most efficient algorithm is the \texttt{LP} relaxation with the TKR rounding scheme. While the worst case complexity of \texttt{LP} relaxation seems to be worse than the \texttt{Greedy}, the \emph{Simplex} approach seems to find the best solution faster than the worst case complexity. The time needed for \texttt{LP}-\texttt{TKR} is at least four orders of magnitude faster than the closest competitors - \texttt{Greedy} and \texttt{LS}. Greedy and LS consume a similar amount of time. While Greedy consumes a fixed amount of time, LS is an iterative algorithm. In this scenario, LS and Greedy have consumed a similar amount of time. In some other experiments, the time varies based on the number of iterations until convergence. \texttt{LP}-\texttt{IRP} is efficient compared to the Greedy. LS and HC approaches by more than a few orders of magnitude. But, \texttt{LP}-\texttt{IRP} is slower than \texttt{LP}-\texttt{TKR} by one order of magnitude. Combining the quality results with these approaches, rounding based approaches seem to provide a good trade-off especially, compared to Greedy, LS, and HC. Moreover, the results of these rounding procedures are comparable to that of the competitors. 

The baseline approach of HC consumes around two orders of magnitude more time compared to Greedy and HC when the number of vertices is 512. Additionally, the improvements of HC over Greedy are negligible, as seen in the previous experiments. For some of the further experiments that look at other parameters, we drop HC from the plots due to the excessive amount of time taken to run the algorithm.

\begin{figure}[t]
  \centering
  \includegraphics[width=0.8\columnwidth]{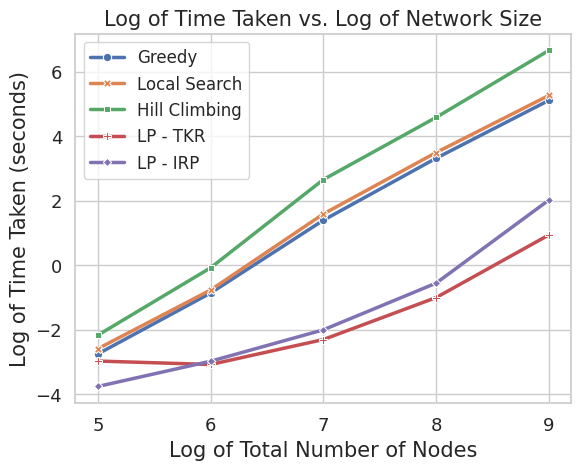}

  \caption{Log (time) vs log (no. of nodes) for Greedy, LS, HC, and LP with rounding for the Gaussian Waxman model}
  \label{fig:time-plot-Waxman}

\end{figure}

\subsubsection{Varying the Number of Topologies}: We explore the effect of varying the number of sampled topologies on the robustness of the result. For the dataset, we generated Gaussian Waxman graphs with 256 and 512 nodes in which 10\% of the nodes were initially infected and the vaccine budget was 10\%. For each network size, the same network was used for all cases to ensure a valid comparison. We sampled 25, 50, 100, 150, 200, 250, and 300 topologies and checked the variance of the number of nodes saved for 10 executions of each algorithm. We skipped ILP and HC due to their large runtimes.

Figures \ref{fig:time-bp_256_LT} and \ref{fig:time-bp_512_LT} show the box plots of the nodes saved for 10 executions of each algorithm, for 256 and 512 nodes respectively. We observe that in general, the average number of nodes saved decreases, which indicates that smaller sample sizes overestimate the effectiveness of a certain set of nodes as vaccines, possibly due to the high variance of the data. As the number of topologies increases, the number of saved nodes decreases but the estimate of the number of saved nodes becomes more reliable.

\begin{figure}[t]
  \centering
  \includegraphics[width=0.9\columnwidth]{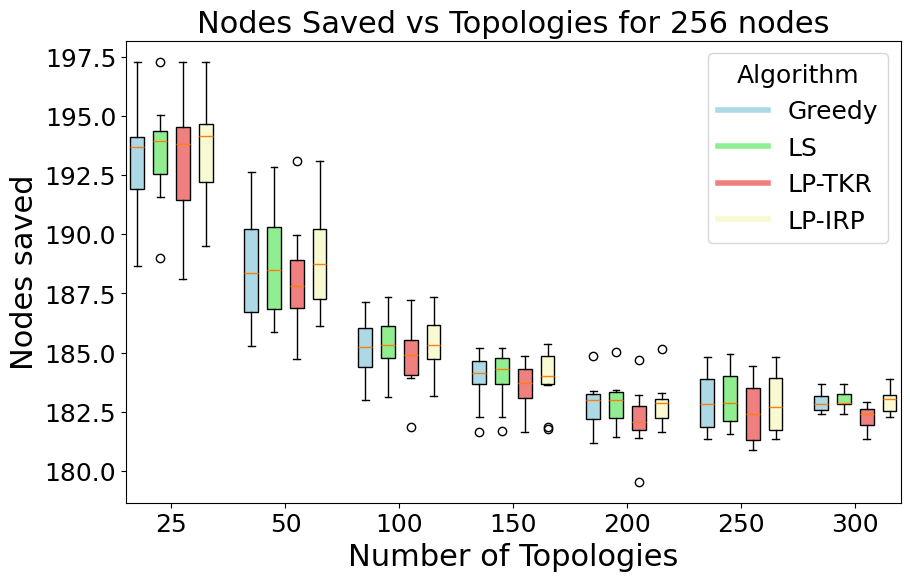}
 
  \caption{For 256 nodes, beyond threshold 200 topologies, the variance of nodes saved reduced}
  \label{fig:time-bp_256_LT}

\end{figure}  

\begin{figure}[ht]
  \centering
   \includegraphics[width=0.9\columnwidth]{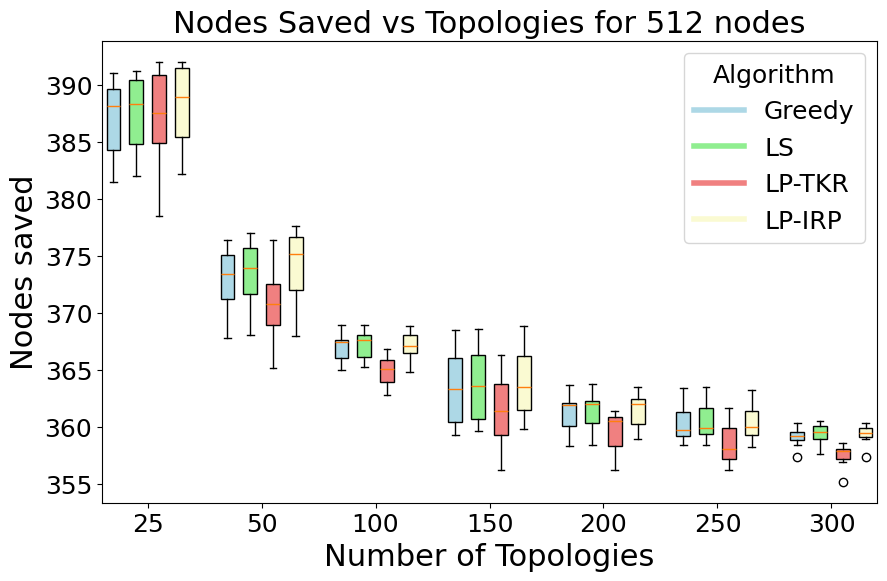}

  \caption{For 512 nodes, beyond threshold 200 topologies, the variance of nodes saved reduced}
  \label{fig:time-bp_512_LT}

\end{figure}

\subsubsection{Quality and runtime Results for Texas and NY Datasets} \label{subsec:quality-runtime-Texas-NY}

In this experiment, we use the Texas dataset which has a total population of 34327175, and create a miniature version for our experiments. We scale the Texas dataset by representing a group of people by a single vertex in a graph. This is akin to the notion of resolution of images, where one could look at the same picture at different resolutions. 
We use two scales, one to have 1 node represent 100000 people (resulting in a graph with 238 vertices) and another to have 1 node represent 50000 people (resulting in a graph with 517 vertices). Towns with a very low population are not used for our modeling. In these miniature graphs, we fix the set of infected nodes and the vaccine budgets at around 10\% of the total vertices and set a sampling budget of 50 topologies.

The results for this experiment are presented in Table~\ref{tab:percentage-saved-Texas}. An interesting observation is that the percentage of vertices saved is similar across the different scaling factors. This clearly shows that our models capture the essence of the disease-spreading mechanism even when the number of nodes is scaled to very low factors, i.e. even at low resolutions. Similar results can also be observed for the NY dataset, the results for which are presented in Table~\ref{tab:percentage-saved-NY}. 

We also measure the time taken by the different methods in the table. The results show that \texttt{LP-TKR} provides a good trade-off between quality and time taken. As seen in the previous sections, Greedy performs slightly better compared to LP relaxation at the cost of runtime. Moreover, \texttt{LP-IRP} performs better in these experiments \texttt{LP-TKR} with slightly more amount of time. Thus, we consider \texttt{LP-TKR} and \texttt{LP-IRP} to be used for practical purposes.

\begin{table}[t]
    \centering
    \resizebox{\linewidth}{!}{%
    \begin{tabular}{|c|cc|cc|cc|cc|}
        \hline
        \multirow{2}{*}{Scaling Factor} & \multicolumn{2}{c|}{Greedy} & \multicolumn{2}{c|}{ILP} & \multicolumn{2}{c|}{LP-TKR} & \multicolumn{2}{c|}{LP-\texttt{IRP}} \\
        & \% Saved & Time (s) & \% Saved & Time (s) & \% Saved & Time (s) & \% Saved & Time (s) \\
        \hline
        100000 & 74.81 & 14.3039 & \textbf{75.25} & 0.4299 & 74.99 & 0.2621 & 75.25 & 0.3789 \\
        50000 & 72.55 & 167.808 & \textbf{72.75} & 758.3711 & 72.09 & 0.9015 & 
72.60 & 2.6465\\
        \hline
    \end{tabular}%
    }
    \caption{Percentage of nodes saved when varying the scale factor for the Texas dataset}\label{tab:percentage-saved-Texas}
\end{table}

\begin{table}[t]
    \centering
    \resizebox{\linewidth}{!}{%
    \begin{tabular}{|c|cc|cc|cc|cc|}
        \hline
        \multirow{2}{*}{Scaling Factor} & \multicolumn{2}{c|}{Greedy} & \multicolumn{2}{c|}{ILP} & \multicolumn{2}{c|}{LP-\texttt{TKR}} & \multicolumn{2}{c|}{LP-\texttt{IRP}} \\
        & \% Saved & Time (s) & \% Saved & Time (s) & \% Saved & Time (s) & \% Saved & Time (s) \\
        \hline
        100000 & 73.38 & 32.2715 & \textbf{73.67} & 3.0190 & 72.85 & 0.4858 & 73.66 & 0.8105 \\
        50000 & 73.51 & 297.648 & \textbf{73.72} & 1097.400 & 72.66 & 1.0876 & 73.17 & 3.0699\\
        \hline
    \end{tabular}%
    }
    \caption{Percentage of nodes saved when varying the scale factor for the New York (state) dataset}\label{tab:percentage-saved-NY}

\end{table}

\subsection{Experimental Results for Independent Cascade}

We perform a similar set of experiments to the linear threshold models on independent cascade models. Unlike the linear threshold samples, the independent cascade samples can have more than one incoming edge per node. Therefore, keeping the same values of parameters $\alpha$ and $\beta$ in the Gaussian Waxman graphs results in overly dense samples. To counter this, we change the graph generation parameters to $\alpha = 0.05$ and $\beta = 0.5$ for all experiments.

\subsubsection{Quality results for Gaussian Waxman graphs}\label{subsubsection:quality-waxman-ICE}

We measure the quality of the results for the different methods proposed in this paper for the Independent Cascade model. We generated graphs with 64, 128, 256, and 512 vertices using the Gaussian Waxman generation method for 5 centers with 10\% of the nodes initially infected which were chosen randomly, and the budget of vaccines being 10\% of the total population. For each of these graphs, we sampled 50 topologies, i.e. $s=50$. 

We measure the quality of the different algorithms and report in Table~\ref{table:avg-nodes-saved-Gaussian-Waxman-ICE}. Since BLP finds the global optimum for the given set of topologies, we compare the quality of other approaches against BLP. However, we did not run BLP for 512 nodes due to its high runtime. Unlike Linear Threshold, the number of saved nodes doesn't increase with an increase in the total number of nodes, which means that the percentage of saved nodes decreases significantly with an increase in network size. This suggests that vaccination with 10\% of nodes is ineffective for the independent cascade model which can be attributed to its sampling technique. 

The \texttt{LP} programs with both rounding techniques outperform the greedy, local search and hill-climbing algorithms for smaller-sized networks, giving solutions close to the optimal; however, for larger networks, the opposite happens. Also, local search and hill-climbing improve the greedy solution by small margins.

\begin{table}[t]
    \centering
    \resizebox{\linewidth}{!}{%
    \begin{tabular}{|c|c|c|c|c|c|c|}
        \hline
        \# Nodes & Greedy & LS & HC & BLP & LP - \texttt{TKR} & LP - \texttt{IRP} \\
        \hline
        64  & 49.4 & 49.4 & 49.4 & \textbf{53.4} & 53.4 & 53.4\\
        128   & 68.96 & 70.64 & 74.76 & \textbf{87.5} & 81.68 & 87.5 \\
        256  & 57.68  & 57.68 & 59.3 & \textbf{59.64} & 44.92 & 46.34\\
        512 & 63.72 & 63.74 & 63.9 & \textbf{--} & 52.4 & 53.48\\
        \hline
    \end{tabular}%
    }
    \caption{Average number of nodes saved by different algorithms for the Gaussian Waxman graph for Independent Cascade}\label{table:avg-nodes-saved-Gaussian-Waxman-ICE}
\end{table}

\subsubsection{Runtime Results for Gaussian Waxman graphs}\label{subsec:runtime-gw-ice}

In this experiment, we measure the runtimes of different algorithms for the Waxman Gaussian graphs of Independent Cascade models. The number of vertices is varied, with 10\% of nodes initially infected and a budget of 10\% vaccines.

Figure \ref{fig:time-plot-Waxman-ICE} presents the results of this experiment. As with the experiment described in \S~\ref{subsec:runtime-waxman} for linear threshold, \texttt{LP} with \texttt{TKR} rounding is the most efficient algorithm. Greedy and LS search algorithms take a similar amount of time. Therefore, local search does not find much improvement over the greedy solution, just like the results obtained for linear threshold. \texttt{LP}-\texttt{IRP} runs faster than the greedy algorithm, but the greedy solution is closer to optimum in some cases, as seen in Table \ref{table:avg-nodes-saved-Gaussian-Waxman-ICE} for \S~\ref{subsec:runtime-waxman}. 

The hill-climbing algorithm consumes around two orders of magnitude more time than the greedy. Moreover, it does not provide much improvement over the greedy solution. We, therefore, do not consider HC for the subsequent experiments.
The BLP, which gives the optimum solution for a given set of topologies, has a very high runtime for larger graphs. For example, for a Waxman Gaussian graph with 256 nodes, the ILP takes 271586 seconds. In comparison, the greedy algorithm takes 63 seconds, and \texttt{LP}-\texttt{TKR} takes less than 1 second. We, therefore, focus on comparing the results obtained by our proposed approaches for further experiments.

\begin{figure}[t]
  \centering
  \includegraphics[width=0.8\columnwidth]{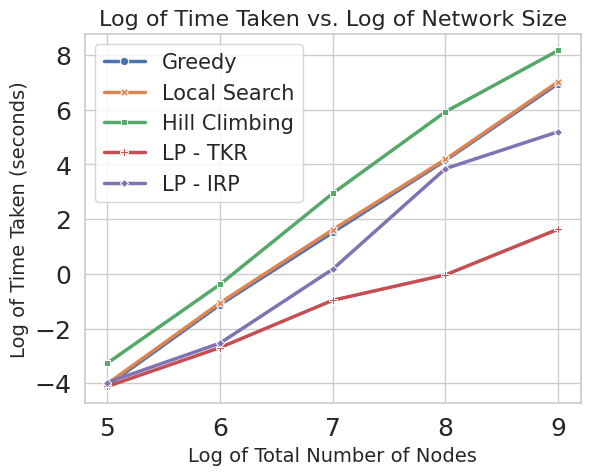}
  % \vspace{-4mm}
  \caption{Log (time) vs log (no of nodes) for Greedy, LS, HC, LP - TKR and LP - IRP methods for the Gaussian Waxman graphs for Independent Cascade}
  \label{fig:time-plot-Waxman-ICE}
  
\end{figure}

\subsubsection{Varying the \% budget of vaccines}
We explore the effect of varying the vaccine budget for datasets of Gaussian Waxman graphs of 512 nodes for 10\% (51) initially infected nodes and 50 topologies. We vary the number of vaccinated vertices to 5\% (25 vertices), 10\% (51 vertices), 20\% (102 vertices), 30\% (138), 40\% (204 vertices), and 50\% (256 vertices). Since HC and BLP are too time-consuming for larger independent cascade graphs, we skip these two methods and analyze the number and percentage of nodes saved by the other approaches.

Our results are presented in Table \ref{table:nodes-saved-budget-vaccines-varied-IC}. An interesting observation is that for each approach, the jump in the number of nodes saved becomes higher with successive increases in the vaccine budget. This behavior is particularly pronounced for \texttt{LP} with \texttt{IRP}, for which increasing the budget from 40\% to 50\% (giving the vaccine to an additional 52 nodes) resulted in around 152 more nodes being saved. Therefore, the law of diminishing returns does not apply for the independent cascade case.

\begin{table}[t]
    \centering
    \resizebox{0.8\linewidth}{!}{%
    \begin{tabular}{|c|c|c|c|c|}
        \hline
        Budget \% & Greedy & LS & LP - \texttt{TKR} & LP - \texttt{IRP}\\
        \hline
        5 & 30.5 & 30.5  & 25.78 & 25.72\\
        10 & 63.72 & 63.74  & 52.4 & 53.48\\
        20 & 125.88 & 126.16  & 106.88 & 106.86\\
        30 & 192.34 & 194.58  & 158.82 & 158.88\\
        40 & 248.68 & 256.6  & 218.72 & 212.7\\
        50 & 340.28 & 358.34  & 310.76 & 364.84\\
        \hline
    \end{tabular}%
    }
    \caption{Average number of nodes saved when varying the vaccination budgets for independent cascade}\label{table:nodes-saved-budget-vaccines-varied-IC}
\end{table}

\subsubsection{Quality and Runtime Results for Texas and NY Datasets}\label{subsec:real-world-ice}
We measure the quality and runtime of our approaches on the New York and Texas datasets for the two scales described in \S~\ref{subsec:quality-runtime-Texas-NY} using the Independent Cascade model. Table \ref{tab:percentage-saved-NY-IC} presents the results for the NY dataset and Table \ref{tab:percentage-saved-Texas-IC} presents the results for the Texas dataset. 

Similar to \S~\ref{subsubsection:quality-waxman-ICE}, the percentages of saved nodes are much lower than the linear threshold counterparts. The percentages of saved nodes are similar across different scales for both datasets, indicating that even scaled versions of a real-world dataset can give insights into the effectiveness of a particular vaccine budget in preventing disease spread.

The tables also show the runtimes of different approaches. \texttt{LP}-\texttt{TKR} is highly efficient and gives results comparable to \texttt{LP}-\texttt{IRP}. Therefore, \texttt{LP}-\texttt{TKR} is a good approach for this setting. The greedy and LS algorithms take more time to run than the \texttt{LP} approaches and have a larger variation in the percentages of nodes saved across different scales than the \texttt{LP}s. Therefore, \texttt{LP}-\texttt{TKR} and \texttt{LP}-\texttt{IRP} are useful approaches for practical purposes.

\begin{table}[t]
    \centering
    \resizebox{\linewidth}{!}{%
    \begin{tabular}{|c|cc|cc|cc|cc|}
        \hline
        \multirow{2}{*}{Scaling Factor} & \multicolumn{2}{c|}{Greedy} & \multicolumn{2}{c|}{LS} & \multicolumn{2}{c|}{LP-TKR} & \multicolumn{2}{c|}{LP-\texttt{IRP}} \\
        & \% Saved & Time (s) & \% Saved & Time (s) & \% Saved & Time (s) & \% Saved & Time (s) \\
        \hline
        100000 & 13.6196 & 752.364 & 13.6196 & 783.9387 & 10.8655 & 15.5875 & 10.8655 & 112.245 \\
        50000 & 10.08 & 13123.9 & 10.08 & 13702.337 & 10.4864 & 61.5389 &  11.52 & 322.8992\\
        \hline
    \end{tabular}%
    }
    \caption{Percentage of nodes saved when varying the scale factor for the New York dataset}\label{tab:percentage-saved-NY-IC}
\end{table}

\begin{table}[t]
    \centering
    \resizebox{\linewidth}{!}{%
    \begin{tabular}{|c|cc|cc|cc|cc|}
        \hline
        \multirow{2}{*}{Scaling Factor} & \multicolumn{2}{c|}{Greedy} & \multicolumn{2}{c|}{LS} & \multicolumn{2}{c|}{LP-\texttt{TKR}} & \multicolumn{2}{c|}{LP-\texttt{IRP}} \\
        & \% Saved & Time (s) & \% Saved & Time (s) & \% Saved & Time (s) & \% Saved & Time (s) \\
        \hline
        100000 & 13.5714 & 251.06 & 13.5714 & 263.3358 & 11.2605 & 9.6391 & 11.2857 & 51.3749 \\
        50000 & 10.6382 & 5326.43 & 10.6382 & 5531.013 & 10.0618 & 42.5712 & 10.2553 & 434.4422\\
        \hline
    \end{tabular}%
    }
    \caption{\% of nodes saved when varying the scale factor for the Texas dataset}\label{tab:percentage-saved-Texas-IC}

\end{table}

\subsection{Summary of results}
We summarize the experimental results below:

\begin{itemize}[leftmargin=*]
    \item (\S~\ref{subsubsection:quality-waxman} and \S~\ref{subsubsection:quality-renyi-erdos}) Among the different approaches, the greedy approach performs well and the results are comparable to Optimal. LS and HC improved the results generated from Greedy marginally. \texttt{LP}-\texttt{TKR} and \texttt{LP}-\texttt{IRP} produces good results, with \texttt{IRP} rounding always outperforming \texttt{TKR}.
    \item (\S~\ref{subsec:runtime-waxman} and \S~\ref{subsec:runtime-gw-ice}) The LP-based rounding methods perform extremely well (many orders of magnitude faster) in practice compared to the competitors. HC is many orders of magnitude slower than  others.
    \item (\S~\ref{subsec:varying-budget-vaccines}) Our meta-experiment with varying the vaccine budget shows that \emph{law of diminishing returns} applies to the vaccination setting for the linear threshold model. Initially, vaccinating critical nodes of the network provides the bang for the buck which seems very intuitive. 
    \item (\S~\ref{subsec:quality-runtime-Texas-NY}) Our real-world experiments with the linear threshold model show that our methods were able to save around 75\% of the population when we have a vaccine budget of 10\%. 
    \item (\S~\ref{subsubsection:quality-waxman-ICE} and \S~\ref{subsec:real-world-ice}) A 10\% vaccination budget is not sufficient for the independent cascade model and saves a very small fraction of total nodes.
    \item As a note on the experiments, \texttt{LP} rounding-based methods provide a good trade-off between time and quality, which would be ideal candidates for practitioners adapting our methods. 
\end{itemize}

\section{Discussion and Future Work}\label{sec:discussion-future-work}

Our empirical analysis illustrates the efficiency and efficacy of the different approaches. While these show good practical applicability, it would be good  to theoretically analyze these techniques and deduce an approximation factor. Additionally, a lower-bound analysis of the problem is an open problem in its own right. Another interesting direction that warrants further research pertains to the alternate models of disease spreading. While we analyze two of the most applicable and practical models of disease spreading, one could consider other models that might be applicable to lesser-studied modes of transmission of viruses, bacteria, and fungi.

\section{Conclusion}\label{sec:conclusion}

In this paper, we initiated research into the problem of disease spreading in real-world networks. More specifically, we defined the problem of finding $k$ important nodes in a graph which need to vaccinated to limit the spread of disease in the network. We proposed a sampling-based approach to overcome enumerating an exponential number of topologies in the network. As an algorithmic contribution, we devised a Linear Program with Binary variables that can solve our problem without any error for a given set of topologies. To further improve the runtime, we propose a range of algorithms which trade-off quality of the results and time. Toward this goal, we proposed an LP relaxation approach with two different rounding mechanisms, an intuitive Greedy approach, followed by a local search heuristic. We backed our algorithms with rigorous runtime analysis and comprehensive experimental evaluations.

% \newpage
\bibliographystyle{ACM-Reference-Format}
\bibliography{sample}

\end{document}